\documentclass[11pt]{article}
\usepackage{amssymb}
\usepackage{amsmath}
\usepackage{amsfonts,mathabx}
\usepackage{amsthm}

\usepackage{comment}
\usepackage[textwidth=2cm]{todonotes}
\usepackage{enumerate}
\usepackage[linesnumbered, ruled, lined, boxed,]{algorithm2e}
\usepackage{authblk}
\usepackage[T1]{fontenc}
\usepackage{listings}
\usepackage{booktabs}
\usepackage[toc,page]{appendix}
\usepackage{setspace}
\usepackage[left=3.5cm,top=3cm,right=3cm,bottom=2.5cm]{geometry}
\usepackage[]{algorithm2e}

\newtheorem{theorem}{Theorem}[section]
\newtheorem{definition}[theorem]{Definition}
\newtheorem{lemma}[theorem]{Lemma}
\newtheorem{corollary}[theorem]{Corollary}

\newcommand{\N}{\mathbb{N}}

\newcommand{\M}{\mathcal{M}}

\renewcommand{\P}{\mathcal{P}}
\newcommand{\R}{\mathbb{R}}
\newcommand{\Q}{\mathbb{Q}}
\newcommand{\G}{\mathcal{G}}
\renewcommand{\H}{\mathcal{H}}
\usepackage[colorlinks,linkcolor=blue,citecolor=blue,urlcolor=blue]{hyperref}

\begin{document}
	\pagestyle{plain}
	\title{Computing Equilibria in Atomic Splittable Polymatroid Congestion Games with Convex Costs}	
	\author[1]{Tobias Harks}
	\author[2]{Veerle Timmermans \thanks{This work is part of the research programme \emph{Optimal Coordination Mechanisms for Distributed Resource Allocation} with project number 617.001.302, which is (partly) financed by the Netherlands Organisation for Scientific Research (NWO). }}
	\affil[1]{Department of Mathematics, Augsburg University}
	\affil[2]{Department of Management Science, RWTH Aachen}	
	\date{\today}
	\maketitle
	\begin{abstract}
In this paper, we compute $\epsilon$-approximate Nash equilibria in atomic splittable polymatroid congestion games with convex Lipschitz continuous cost functions. The main approach relies on computing a pure Nash equilibrium for an associated \emph{integrally-splittable} congestion game, where players can only split their demand in integral multiples of a common packet size. It is known that one can compute pure Nash equilibria for integrally-splittable congestion games within a running time that is pseudo-polynomial in the aggregated demand of the players. As the main contribution of this 
paper, we decide for every $\epsilon>0$, a packet size $k_{\epsilon}$ and prove that the associated $k_{\epsilon}$-splittable Nash equilibrium is an $\epsilon$-approximate Nash equilibrium for the original game. 

We further show that our result applies to multimarket oligopolies with decreasing, concave Lipschitz continuous price functions and quadratic production costs: there is a polynomial time transformation to atomic splittable polymatroid congestion games implying that we can compute $\epsilon$-approximate Cournot-Nash equilibria  within pseudo-polynomial time.
	\end{abstract}


\section{Introduction}
\label{sec:introduction} 
Congestion games are a fundamental problem class in the operations research 
and algorithmic game theory literature. 
They model equilibrium problems in transportation~\cite{Haurie85}, communication~\cite{Korilis1997,Orda93} or logistics systems~\cite{Cominetti09}
and  from a theoretical point of view, they are frequently
used as benchmark models for new concepts or algorithms.
While by now, several results regarding the existence and computability of equilibria
are known (cf.~\cite{Ackermann08,CaragiannisFGS15,Gairing13,Harks:existence,Skopalik08}), there is one particular subclass of congestion games, namely \emph{atomic splittable congestion games},  that is not well understood.   In such game, 
there is set of resources and a set of players, where  each
 player is associated with a positive demand. In a  simple symmetric version, a strategy of every player is a (possibly fractional) distribution of the player-specific demand over the resources. Resources have load-dependent, non-decreasing, and convex costs and
the total private cost of a player is simply the sum of player's costs over the resources. 
Even for this seemingly simple model, there are currently no efficient algorithms computing  (approximate) pure Nash equilibria known. \footnote{See Definition~\ref{def:approx-eq} for the term ``$\epsilon$-approximate equilibrium''.}

\subsection{Related Work}
Very recently (and independently of our work) Bhaskar and Lolakapuri~\cite{bhaskar18} devised two algorithms computing an $\epsilon$-approximate equilibrium. Both algorithms run in exponential time in either the number of resources or the number of players, respectively. 
They are are based on  guessing the \emph{marginal costs} of the players at an equilibrium. These marginal cost appear to have several monotonicity properties, which they exploit using a high-dimensional binary search algorithm. 

Harks and Timmermans~\cite{TimmermansIPCO2017} developed a  polynomial time algorithm that computes exact Nash equilibria when the cost functions are affine. Their idea is to compute a pure Nash equilibrium for an associated \emph{integrally splittable} congestion game. In such games, players can only split their demand in integral multiples of a common \emph{packet size}. 
The class of integrally splittable singleton congestion games has been studied before by Tran-Thanh et al.~\cite{Tran11} for the case of player-independent convex cost functions. Later, this problem was studied by Harks et al.~\cite{harks2014resource,harks2016resource}  for the more general
case of polymatroid strategy spaces and player-specific convex cost functions. In particular, Harks et al. proved that the algorithm by Tran-Thanh et al. has a running time which is pseudo-polynomial in the aggregated load of the players (cf. Corollary 5.2~\cite{harks2016resource}).

Deligkas et al. \cite{Deligkas2016} study the computation of $\epsilon$-approximate equilibria in general concave games with compact strategy spaces and Lipschitz continuous cost functions. In their paper, they decide on a number $k$, discretize the strategy space and only consider $k$-uniform points, i.e., vectors where all elements are integer multiples of $k$. Then, as for each player only finitely many of these vectors exist, they enumerate all feasible $k$-uniform strategy profiles, and pick the best candidate (see also Lipton et al.~\cite{LiptonMM03} for a similar approach). This method results in an algorithm that finds $\epsilon$-approximate equilibria in exponential time.

\subsection{Our Results}
In this paper, we compute $\epsilon$-approximate Nash equilibria for atomic splittable \emph{polymatroid} congestion games with non-decreasing, convex and Lipschitz continuous cost functions. In such a game, the strategy space of every player is polymatroid base polytope
which strictly generalizes the singleton setting  considered by Bhaskar and Lolakapuri~\cite{bhaskar18}.
 Our main approach --  similar to the approach in~\cite{TimmermansIPCO2017} -- relies on computing a pure Nash equilibrium for an associated integrally splittable congestion game.
We prove that for each $\epsilon>0$, we can compute a packet size $k_{\epsilon}$ that such that the $k_{\epsilon}$-integral equilibrium is guaranteed to be an $\epsilon$-approximate equilibrium. Thus, this implies that one can compute an $\epsilon$-approximate equilibrium within pseudo-polynomial time by using the algorithm by Harks et al.~\cite{harks2014resource,harks2016resource}. 
In contrast to the algorithms of Bhaskar and Lolakapuri~\cite{bhaskar18}, our algorithm works
even for player-specific cost functions.

We further show that our result applies to multimarket oligopolies  (cf. Bulow~\cite{Bulow85}) with decreasing, concave Lipschitz continuous price functions and quadratic production costs. We give a polynomial time transformation to atomic splittable congestions games with player-specific convex costs, thus,  our algorithm and the transformation can be used to compute  $\epsilon$-approximate equilibria for multimarket oligopolies.

\section{Preliminaries}
\label{sec:preliminaries}


\paragraph{Polymatroids.}
Let $E=\{e_1,\dots, e_m\}$ be a finite set of resources and  $\rho:2^E\rightarrow\R $ be  (1) submodular, i.e, $\rho(U)+\rho(V) \geq \rho(U \cup V) + \rho(U \cap V)$ for all $U, V \subseteq E$, (2)
 monotone, i.e., $\rho(U) \leq \rho(V)$ for all $U \subseteq V$ and (3) normalized, i.e, $\rho(\emptyset) = 0$.  Then, the pair $(E,\rho)$ is called a \emph{polymatroid} and  the associated polyhedron is defined as:
\[ \P_{\rho} := \left\{  x \in \R_{\geq 0}^{E} \mid x(U) \leq \rho(U) 
\; \forall U \subseteq E \; \right\},\]
where $x(U) := \sum_{e \in U} x_e$ for all $U\subseteq E$. Given a polyhedron $\P_{\rho}$ and a rational $d \in \Q_{>0}$ with $d \leq x(E)$, a polymatroid base polytope of rank $d$ is defined as: 
\[ \P_{\rho}(d) := \left\{  x \in \R_{\geq 0}^{E} \mid x(U) \leq \rho(U) 
\; \forall U \subseteq E, \; x(E) = d \right\}.\]

\paragraph{Atomic splittable polymatroid congestion games.}
An atomic splittable polymatroid congestion game is represented  by the tuple:
$
\G:= \left(N,E,(d_i)_{i \in N}, (\rho_i)_{i \in N},(c_{i,e})_{i \in N, e \in E} \right).
$
Here,  $N=\{1,\dots, n\}$ is a finite player set and with every $i\in N$, we associate a player-specific polymatroid $(E,\rho_i)$. The strategy space of player $i\in N$ is defined as the (player-specific)  polymatroid base polytope 
\[ \P_{\rho_i}(d_i):= \left\{  x_i \in \R_{\geq 0}^{E} \mid x_i(U) \leq \rho_i(U) \; \forall U \subseteq E, \; x_i(E) = d_i\right\}.\]
The combined strategy space is denoted by $\P := \prod _{i \in N} \P_{\rho_i}(d_i)$ and we denote by ${x}=({x}_i)_{i\in N}$ the overall strategy profile. The entry $x_{i,e}$ of the vector $x_i$  is the load of player $i$ on $e \in E$ and  $x_e:=\sum_{i\in N}x_{i,e}$ is defined as the total load on $e$. Resources have player-specific cost functions $c_{i,e}(x_e)$, where $c_{i,e}(x_e)$ is non-decreasing, non-negative, differentiable and convex. We further  assume that all cost functions $c_{i,e}(x)$ and their derivatives $c'_{i,e}(x)$ are Lipschitz continuous:
\begin{definition}[Lipschitz continuity]
	A function $c: \mathbb{R} \rightarrow \mathbb{R}$ is called Lipschitz continuous, if there exists a constant $L > 0$ such that for all $x,y \in R$:
	$|c(y)-c(x)| \leq L|y-x|. $
	Here $L$ is called the Lipschitz constant.
\end{definition}

The total cost of player $i$ in strategy distribution ${x}$ is defined as
$\pi_{i}(x)=\sum_{e\in E} c_{i,e}(x_e)\,x_{i,e}.$
The goal of each player is to choose a strategy $x_i$ such that her personal cost $\pi_i(x)$ is minimized.
As pure Nash equilibria in these games are not guaranteed to be rational,  the notion of $\epsilon$-approximate equilibria are suitable.
\begin{definition}[$\epsilon$-approximate equilibrium]\label{def:approx-eq}
	Strategy $x$ is an $\epsilon$-approximate equilibrium for game $\G$,  if for each $i \in N$, and every strategy $y_i \in \mathcal{P}_i$: 
	$\pi_i(y_i,x_{-i}) \geq \pi_i(x_i,x_{-i}) - \epsilon. $
\end{definition}
A pair $\bigl(x,(y_i,x_{-i})\bigr) \in \P\times \P$ is called an \emph{improving move} of player $i$, if $\pi_i(x_i,x_{-i}) > \pi_i(y_i,x_{-i})$. Given $x_{-i} \in \P_{-i}(d_{-i})$, a strategy $x_i \in \P_i(d_i)$  is called a \emph{best response} of player $i$ to $x_{-i}$ if $\pi_i(x_i,x_{-i}) \leq \pi_i(y_i,x_{-i})$ for all $y_i \in \P_i(d_i)$. 

When $E, N,(\rho_i)_{i \in N}$ and $(c_{i,e})_{i \in N, e \in E}$ are clear from the context, we refer to the game as $\G((d_i)_{i \in N})$, and write $\P_i(d_i)$ instead of $\P_{\rho_i}(d_i)$. For each $i \in N$, we
write $\P_{-i}(d_{-i}) = \prod_{j \neq i} \P_{j}(d_j)$ and $x = (x_i,x_{-i})$ meaning that $x_i \in \P_i(d_i)$ and $x_{-i} \in \P_{-i}(d_{-i})$.

\paragraph{Integral Polymatroid Congestion Games.}
\label{sec:integralsplittablecongestiongames}
A $k$-integral polymatroid congestion game is given by the tuple $\G_k:= \left(N,E,(d_i)_{i \in N}, (\rho_i)_{i \in N}, (c_{i,e})_{i \in N, e \in E},k \right)$ with $k\in \Q_{>0}$ such that $d_i/k \in\N$ for all $i\in N$. Here, players cannot choose any strategy in the polymatroid base polytope, but only $k$-integral points. Thus, the strategy space of player $i$ is defined by the (player-specific) $k$-integral polymatroid base polytope:
\[ P^k_{\rho_i}(d_i):= \left\{  x_i \in \R_{\geq 0}^{E} \mid x_i(U) \leq \rho_i(U) \; \forall U \subseteq E, \; x_i(E) = d_i \text{ and } x_{i,e} =kq, q \in \N_{\geq 0}\right\}.\]
When $E, N,(\rho_i)_{i \in N}$ and $(c_{i,e})_{i \in N, e \in E}$ are clear from the context, we refer to the game as $\G_k((d_i)_{i \in N})$ and to the strategy spaces as  $P^k_i(d_i)$. Similar to atomic splittable congestion games, the complete strategy space of the game is defined by $\P^k := \prod_{i \in N} \P^k_i(d_i)$. For $i \in N$, we write $\P^k_{-i}(d_{-i}) = \prod_{j \neq i} \P^k_{j}(d_j)$ and $x = (x_i,x_{-i})$ meaning that $x_i \in \P^k_i(d_i)$ and $x_{-i} \in \P^k_{-i}(d_{-i})$. A strategy profile $x$ is an \emph{equilibrium} if $\pi_i(x) \leq \pi_i(y_i, x_{-i})$ for all $i \in N$ and $y_i \in \P^k_i(d_i)$. A pair $\bigl(x,(y_i,x_{-i})\bigr) \in \P^k\times \P^k$ is called an \emph{improving move} of player $i$, if $\pi_i(x_i,x_{-i}) > \pi_i(y_i,x_{-i})$. Given $x_{-i} \in \P^k_{-i}(d_{-i})$, a strategy $x_i \in \P^k_i(d_i)$  is called a \emph{best response} of player $i$ to $x_{-i}$ if $\pi_i(x_i,x_{-i}) \leq \pi_i(y_i,x_{-i})$ for all $y_i \in \P^k_i(d_i)$.

Harks et al.~\cite{harks2016resource} developed an algorithm that computes an exact Nash equilibrium for $k$-integral splittable polymatroid congestion games with non-negative, increasing and convex cost functions. The running time of this algorithm is pseudo-polynomial in the aggregated demand of the players. Their algorithm~\cite[Algorithm 1]{harks2016resource} starts with an equilibrium for the game where the demand for each player is set to zero: $d'_i=0$. This game has a unique equilibrium, where $x_{i,e}=0$ for each $i \in N$ and $e \in E$. Then, they repeatedly look for a player for who $d'_i<d_i$. For this player, they increase $d'_i$ by $k$ and a preliminary equilibrium with respect to the current demands $d'_i$ is recomputed by following a sequence of best responses of the players. The running time of~\cite[Algorithm 1]{harks2016resource} is $O(nm (d/k)^3)$, where $d$ is an upper bound on the demands $d_i$.

\section{Lipschitz Continuity and Approximate Equilibria}
\label{sec:lipschitz continuity and approximate equilibria}
As our main result, we prove that when all cost functions and their derivatives are \emph{Lipschitz continuous} with Lipschitz constant $L$, we can compute a packet size $k_{\epsilon}$ such that the exact $k_{\epsilon}$-integral equilibrium is an $\epsilon$-approximate atomic splittable equilibrium. Combined with~\cite[Algorithm 1]{harks2016resource}, this gives a pseudo-polynomial time algorithm computing an $\epsilon$-approximate equilibrium for atomic splittable polymatroid congestion games with convex cost functions.

For preparing the necessary technical tools, in Section~\ref{todo} we first discuss a sufficient and necessary conditions for a strategy $x$ to be either an atomic splittable equilibrium or a $k$-integral equilibrium. Then, in Section~\ref{sec:epsilonapproximateequilibrium}, we use these conditions to prove that for any $\epsilon>0$, we can find a packet size $k_{\epsilon}$ such that the exact $k_{\epsilon}$-integral equilibrium is an $\epsilon$-approximate atomic splittable equilibrium.

\subsection{Equilibrium Conditions in $k$-Integral Games}
\label{todo}
In atomic splittable congestion games, the \emph{marginal cost} for player $i$ on resource $e$ is defined as: $\mu_{i,e} (x) = c_{i,e}(x_e) + x_{i,e}c'_{i,e}(x_e).$
Intuitively, it represents the cost increase of player $i$ when she would increase her load on resource $e$. For $k$-integral games, the marginal costs are defined as follows: 
\begin{align} \label{marginalcostup}
\mu_{i,e}^{+k} (x) &= (x_{i,e}+k)c_{i,e}(x_e+k) - x_{i,e}c_{i,e}(x_e), \\
\mu_{i,e}^{-k} (x) &=
\begin{cases}
x_{i,e}c_{i,e}(x_e) -  (x_{i,e}-k)c_{i,e}(x_e-k),
& \text{if } x_{i,e} >  0,\\
-\infty,
& \text{if } x_{i,e}\leq 0.
\end{cases}
\end{align}
Here, $\mu_{i,e}^{+k} (x)$ is the cost for player $i$ to add one packet of size $k$ to resource $e$ and $\mu_{i,e}^{-k} (x)$ is the gain for player $i$ for removing a packet of size $k$ from resource $e$. When we apply Theorem 8.1~\cite{Fujishige05} to $k$-integral games, we obtain the following lemma: 

\begin{lemma}\label{lem:discreteexchange}
	Given an atomic splittable congestion game $\G$, and a packet size $k$ such that $\rho_i(U)/k\in \N$ for all $U \subseteq E$ and all $i \in N$. Given a strategy $x_{-i} \in \P^k_{-i}(d_{-i})$ for $k$-integral splittable game $\G_k$, then strategy $x_i \in P^k_i(d_i)$ is a best response for player $i$ if and only if for every pair $(e,f) \in E^2$ the following holds. If there exist an $\alpha >0$ such that:
	$x_i+\alpha(\chi_e-\chi_{f}) \in \P_i(d_i), $
	then we have:
	$\mu^{+k}_{i,e}(x) \geq \mu^{-k}_{i,f}(x).$
\end{lemma}
\begin{proof}
	Assume there exists an $\alpha>0$ such that: $x_i+\alpha(\chi_e-\chi_{f}) \in \P_i(d_i).$ Then, as $\rho_i(U)/k\in \N$ for all $U \subseteq E$ and all $i \in N$, we can define $\alpha'= k \lceil \tfrac{1}{k}\alpha \rceil$, for which holds that $x_i+\alpha'(\chi_e-\chi_{f}) \in \P^k_i(d_i)$. Then, Theorem 8.1~\cite{Fujishige05} implies that $\mu^{+k}_{i,e}(x) \geq \mu^{-k}_{i,f}(x).$
\end{proof}

When the cost functions $c_{i,e}(x)$ and their derivatives $c'_{i,e}(x)$ are Lipschitz continuous with Lipschitz constant $L$, we obtain the following relations between marginal costs in atomic splittable games and $k$-integral games:
\begin{lemma}\label{marginalbound}
	For any feasible strategy $x$ for game $\G_k$, we have for each $i \in N$ and $e \in E$:
$ \tfrac{1}{k}\mu^{+k}_{i,e}(x) \leq \mu_{i,e}(x) + kL (d_i+1)\text{ and }
		\tfrac{1}{k}\mu^{-k}_{i,e}(x) \geq \mu_{i,e}(x) - k L (d_i+1) \text{ whenever } x_{i,e}>0$.

\end{lemma}
\begin{proof} 
	We start by proving $\mu^{+k}_{i,e}(x) /k \leq \mu_{i,e}(x) + kL (d_i+1)$. We obtain:
	\begin{eqnarray*}
		\mu^{+k}_{i,e}(x) / k &=&\frac{ (x_{i,e}+k)c_{i,e}(x_{e}+k) - x_{i,e}c_{i,e}(x_e) }{k} \\
		& = &  x_{i,e} \frac{c_{i,e}(x_{e}+k)-c_{i,e}(x_{e})}{k} + c_{i,e}(x_{e}+k)\\
		& \leq_1 & x_{i,e} c'_{i,e}(x_e + k) + c_{i,e}(x_{e}+k)\\
		& \leq_2 & x_{i,e} (c'_{i,e}(x_e) + kL) + c_{i,e}(x_e) + kL\\
		& = & x_{i,e} c'_{i,e}(x_e) + c_{i,e}(x_e) + (x_{i,e}+1)kL\\
		& \leq &  \mu_{i,e}(x) + kL (\delta+1), \\
	\end{eqnarray*}
where $\delta:=\max_{i \in N} \{d_i\}$. Inequality $\leq_1$ holds as $c_{i,e}$ is convex and increasing, which implies $c'_{i,e}(x_e + k) \geq (c_{i,e}(x_{e}+k)-c_{i,e}(x_{e}))/k$. We obtain inequality $\leq_2$ using Lipschitz constant $L$, as it implies $c_{i,e}(x_e + k)  \leq c_{i,e}(x_e) + kL$ and $ c'_{i,e}(x_e + k)  \leq c'_{i,e}(x_e) + kL$. The second inequality can be obtained in a similar way.
\end{proof}

\begin{lemma} \label{lem:exchangevalue}
	Given an atomic splittable congestion game $\G$, and a packet size $k$ such that $\rho_i(U)/k\in \N$ for all $U \subseteq E$ and all $i \in N$. Let $x_i \in \P^k_i(d_i)$ be a strategy for player $i$ in a $k$-integral game. Then, if there exists an $\alpha >0$ such that
	$x_i + \alpha (\chi_f-\chi_e) \in \P_i(d_i), $
	it holds that
	$\mu_{i,f}(x_k) - \mu_{i,e}(x_k) \geq - 2k_{\epsilon}L (\delta+1).$
\end{lemma}
\begin{proof}
	Lemma~\ref{lem:discreteexchange} states that $\mu^{-k}_{i,e}(x_k) \leq \mu^{+k}_{i,f}(x_k).$
	As $x_i + \alpha (\chi_f-\chi_e) \in \P_i(d_i),$ we have $x_{i,e}>0$. We combine this with Lemma~\ref{marginalbound} to obtain that:
	$$\mu_{i,e}(x_k) - k L (\delta+1) \leq \tfrac{1}{k} \mu^{-k}_{i,e}(x_k) \leq \tfrac{1}{k} \mu^{+k}_{i,f}(x_k) \leq\mu_{i,f}(x_k) + kL (\delta+1).$$
	We rewrite the previous inequality, and obtain the desired statement:
	$$\mu_{i,f}(x_k) - \mu_{i,e}(x_k) \geq - 2kL (\delta+1).$$	
\end{proof}

\subsection{Atomic Splittable and $k$-Integral Equilibria}
\label{sec:epsilonapproximateequilibrium}
In this section we focus on finding $\epsilon$-approximate equilibria for atomic splittable congestion games with increasing, non-negative, differentiable and convex cost functions, where both the original function and its derivative are bounded by Lipschitz constant $L$. In order to do so,  for each $\epsilon>0$ we define a $k_{\epsilon}$, and prove that the $k_{\epsilon}$-integral equilibrium will be an $\epsilon$-approximate atomic splittable equilibrium.

\begin{theorem} \label{thm:suffpacketsize}
	Given an atomic splittable game $\G$, where all cost functions and their derivatives are bounded by a Lipschitz constant $L$. Then, for any $\epsilon>0$, there exists a $k_{\epsilon}>0$ such that an exact equilibrium $x$ for the $k_{\epsilon}$-integral splittable game $\G_{k_{\epsilon}}$ is an $\epsilon$-equilibrium for $\G$.
\end{theorem}
\begin{proof}
	As $\pi_i$ is convex, for any alternative strategy $y_i \in \P_i(d_i)$, we have that:
	\begin{equation} \label{varineq}\pi_i(y_i,x_{-i}) \geq \pi_i(x_i,x_{-i}) + \nabla_i \pi_i(x_i,x_{-i})\cdot (y_i - x_i). \end{equation}
	Thus, our goal is to determine a $k_{\epsilon}$ that bounds 
	$\nabla_i \pi_i(x_i,x_{-i})\cdot (y_i - x_i)$ from below by $\epsilon$. 
	We define $E^{i,+}:=\{e \in E | y_{i,e} >x_{i,e} \} $ and $E^{i,-}:=\{e \in E | x_{i,e} > y_{i,e} \}$.  Note that:
	\begin{equation} \label{eq1}
	\nabla_i \pi_i  (x_i,x_{-i})\cdot (y_i - x_i)
	=  \sum_{e \in E^{i,+}} \mu_{i,e}(x)(y_{i,e} - x_{i,e})+\sum_{e \in E^{i,-}} \mu_{i,e}(x)(y_{i,e} - x_{i,e}).   
	\end{equation}

	Consider the complete, directed, bipartite graph $G(x_i,y_i)$ on node sets $E^{i,-}$ and $E^{i,+}$, where each node $e \in E^{i,-}$ has a supply of $x_{i,e} - y_{i,e}$ and each node $e \in E^{i,+}$ has a demand of $y_{i,e} - x_{i,e}$. The edges of $G(x_i,y_i)$ are directed from $E^{i,-}$ to $E^{i,+}$ and the capacity $c_{e,f}(x_i,y_i)$ of edge $e,f \in  E^{i,-} \times E^{i,+}$ is defined as:
	$$c_{e,f}(x_i,y_i) = \max\{\alpha| x_i+\alpha(\chi_{f} - \chi_{e}) \in \P_i(d_i) \}.$$	

	Then, as $y_i$ and $x_i$ are points in the polymatroid base polytope $\P_i(d_i)$, there exists a transshipment  $t$ in $G(x_i,y_i)$ from resources in $E^{i,-}$ to resources in $E^{i,+}$ that exactly satisfies all supplies, demands and capacities (Lemma 3.2~\cite{Timmermans2017}).
	We denote by $t_{e,f}$ the amount of load transshipped from resource $e$ to resource $f$ in $t$, thus: $\sum_{f \in E^{i,+}} t_{e,f} = x_{i,e}-y_{i,e}$ if $e \in E^{i,-}$ and $\sum_{f \in E^{i,-}} t_{f,e} = y_{i,e}-x_{i,e}$ if $e \in E^{i,+}$
	Using this transshipment, we rewrite~\eqref{eq1} in terms of $t$. Hence,
	
\begin{align} 
	&\nabla_i \pi_i  (x_i,x_{-i})\cdot (y_i - x_i) \nonumber \\
	&=  \sum_{e \in E^{i,+}} \mu_{i,e}(x)(y_{i,e} - x_{i,e})+\sum_{e \in E^{i,-}} \mu_{i,e}(x)(y_{i,e} - x_{i,e}) \nonumber	\\	
	&=  \sum_{e \in E^{i,+}} \mu_{i,e}(x)\left(\sum_{f \in E^{i,-}} t_{f,e} \right)-\sum_{e \in E^{i,-}} \mu_{i,e}(x)\left(\sum_{f \in E^{i,+}} t_{e,f} \right) \nonumber	\\	
&=  \sum_{f,e \in E^{i,-} \times E^{i,+}} \mu_{i,e}(x) t_{f,e}-\sum_{e,f \in E^{i,-} \times E^{i,+}} \mu_{i,e}(x)t_{e,f} \nonumber	\\	
&=  \sum_{e,f \in E^{i,-} \times E^{i,+}} \mu_{i,f}(x) t_{e,f}-\sum_{e,f \in E^{i,-} \times E^{i,+}} \mu_{i,e}(x)t_{e,f} \nonumber	\\	
&=  \sum_{e,f \in E^{i,-} \times E^{i,+}} (\mu_{i,f}(x) - \mu_{i,e}(x))t_{e,f}. \label{transshipment}
	\end{align}	
	
	Note that, in order to use Lemma~\ref{lem:exchangevalue}, we need a packet size $k$ such that $\rho_i(U)/k\in \N$ for all $U \subseteq E$ and all $i \in N$. Note that $\rho_i(U) \in \Q_{\geq 0}$ for all $U \subseteq E$ and all $i \in N$, hence, we define 
	$$\rho_{\gcd}:=\max\{a \in \Q_{>0} \mid a \leq 1 \text{ and }  \forall i \in N, U \subseteq E, \; \exists \ell \in \N \text{ such that } \rho_i(U)=a\cdot\ell \}.$$
	
	Given any $\epsilon>0$, we define 
	\begin{equation} \label{eq:defk} k_{\epsilon}=\frac{\rho_{\gcd}}{\left\lceil \frac{2m^2L\delta(\delta +1)}{\epsilon} \right\rceil}.\end{equation}
	Note that $k_{\epsilon}$ has the following two properties: (1) As $\rho_{\gcd}\leq 1$, we know that $k_{\epsilon} \leq \frac{\epsilon}{2m^2L\delta(\delta +1)}$; (2) As $\rho_{\gcd}/k_{\epsilon} \in \N$, we know that $\rho_i(U)/k_{\epsilon}\in \N$ for all $U \subseteq E$ and all $i \in N$.
	
	We prove that the $k_{\epsilon}$-integral equilibrium is also an $\epsilon$-approximate equilibrium for the corresponding atomic splittable game. Using Lemma~\ref{lem:exchangevalue}, we know that if there exists an $\alpha >0$ such that:
	$x_i + \alpha (\chi_{f}-\chi_e) \in \P_i(d_i), $ we have that:
	\begin{equation} \label{eq:marginalcompare} 
		\mu_{i,f}(x) - \mu_{i,e}(x) \geq  - 2k_{\epsilon}L (\delta+1) \geq -\frac{\epsilon}{m^2 \delta}.
	\end{equation}
	
	By the choice of transshipment $t$, we have that 
	$x_i + t_{e,f} (\chi_{f}-\chi_e) \in \P_i(d_i).$ We combine Equation~\eqref{eq:marginalcompare} with Equation~\eqref{transshipment} and obtain: 
	\begin{equation}  \nabla_i \pi_i(x_i,x_{-i})\cdot (y_i - x_i)  \geq -\left( \sum_{e,f \in E^{i,-} \times E^{i,+}} \frac{\epsilon}{m^2 \delta} t_{e,f} \right). \end{equation}
	Note that $t_{e,f} \leq \delta$ and $|E^{i,-} \times E^{i,+}|< m^2$. Hence:
	$ \nabla_i \pi_i(x_i,x_{-i})\cdot (y_i - x_i) > -\epsilon.$	
	Using Equation~\eqref{varineq} we obtain:
	$\pi_i(y_i,x_{-i}) > \pi_i(x_i,x_{-i}) - \epsilon.$
	Thus, player $i$ cannot gain more than $\epsilon$ by playing an alternative strategy $y_i$. As player $i$ was chosen arbitrarily, $x$ is an $\epsilon$-approximate equilibrium.
\end{proof}

\begin{corollary} Given an atomic splittable singleton game, where the cost functions are non-negative, increasing, differentiable, convex, and where both the original function as its derivative are bounded by a Lipschitz constant, we can compute an $\epsilon$-approximate equilibrium within a running time $$O\left(nm\left(\frac{\delta}{\rho_{\gcd}} \left\lceil \frac{2m^2L\delta(\delta +1)}{\epsilon} \right\rceil\right)^3\right).$$
\end{corollary}
\begin{proof}
	Assume we are given $\epsilon>0$. Using Theorem~\ref{thm:suffpacketsize}, we can then find a packet size $k_{\epsilon}$ such that for any $k \leq k_{\epsilon}$, any $k$-splittable equilibrium is an $\epsilon$-approximate equilibrium. Using the Algorithm~\cite[Algorithm 1]{harks2014resource} by Harks, Peis and Klimm, we can compute a $k$-splittable equilibrium within running time $O(nm(\delta/k)^3)$. Thus, using the definition of $k_{\epsilon}$ in~\eqref{eq:defk}, we can find $\epsilon$-approximate equilibria within the required running time. 
\end{proof}


\section{Multimarket Cournot Oligopoly}
\label{sec:cournot}
In this last section, we derive a strong connection between atomic splittable singleton congestion games with convex cost functions and multimarket Cournot oligopolies with concave, decreasing and differentiable price functions and quadratic costs. Such a game is compactly represented by the tuple
$\mathcal{M}=(N,E,(E_i)_{i \in N}, (p_{i,e})_{i \in N,e \in E_i}, (C_i)_{i \in N}),$ where $N$ is a set of $n$ firms and $E$ a set of $m$ markets. Each firm $i$ only has access to a subset $E_i \subseteq E$ of the markets and each market $e$ is endowed with  firm-specific, non-increasing, differentiable and concave price functions $p_{i,e}(t): \R \rightarrow \R$, for all $ i\in N$. In a strategy profile, a firm $i \in N$ chooses a non-negative production quantity $x_{i,e} \in \R_{\geq 0}$ for each market $e \in E_i$. We denote a strategy profile for a firm by $x_i=(x_{i,e})_{e \in E_i}$, and a joint strategy profile by $x=(x_i)_{i \in N}$. The production costs of a firm are of the form $C_i(t)=c_it^2$ for some $c_i \geq 0$. The goal of each firm $i \in N$ is to maximize its utility, which is given by:

$$u_i(x)=\sum_{e \in E_i} p_{i,e}(x_e)x_{i,e} -C_i\Big(\sum_{e \in E_i}x_{i,e}\Big), $$

where $x_e:=\sum_{i \in N}x_{i,e}$. Note that a connection between Cournot games with affine price functions and atomic splittable games with affine cost functions has already been made in~\cite{TimmermansHarks2016}. In the rest of this section we generalize the connection stated in~\cite{TimmermansHarks2016} and prove that several results that hold for atomic splittable equilibria and $k$-splittable equilibria in games with convex cost functions carry over to multimarket oligopolies with concave price functions.

Both multimarket oligopolies and atomic splittable congestion games are \emph{strategic games}. A strategic game $\G=(N,(X)_{i \in N},(u_i)_i \in N)$ is defined by a set of players $N$, a set of feasible strategies $X_i$ for each player $i \in N$ and a pay-off function $u_i(x)$ for each $i \in N$, where $x \in \bigtimes_{i \in N} X_i$. Let $\G=(N,(X_i)_{i \in N},(u_i)_i \in N)$ and $\H=(N,(Y_i)_{i \in N},(v_i)_i \in N)$ be two strategic games with identical player set $N$. Then, $\G$ and $\H$ are called \emph{isomorphic}, if for all $i \in N$ there exists a bijective function $\phi_i:X_i \rightarrow Y_i$ and $A_i \in \R$ such that $u_i(x_1, \dots x_n) = \nu_i(\phi_1(x_1), \dots ,\phi_n(x_n)) +A_i.$ Then, $(x_i)_{i \in N}$ is an equilibrium of game $\G$ if and only if $(\phi_i(x_i))_{i \in N}$ is an equilibrium of game $\H$. This implies that $(x_i)_{i \in N}$ is the unique equilibrium of game $\G$ if and only if $(\phi_i(x_i))_{i \in N}$ is the unique equilibrium of game $\H$.

\begin{theorem}\label{thrm:contcournot}
	Given a multimarket oligopoly $\M$  with concave, decreasing and differentiable price functions and quadratic costs, one can construct an atomic splittable game $\G$ with convex, increasing and differentiable costs that is isomorphic to $\M$ within polynomial time.
\end{theorem}
\begin{proof}
This proof generalizes a similar transformation stated in~\cite[Theorem 7.2]{TimmermansHarks2016}. Given multimarket oligopoly $\M$, we construct an atomic splittable singleton game $\G$ as follows. For every firm $i \in N$ we create a player $i$ and we define the demand $d_i$ for this player as an upper bound on the maximal quantity that firm $i$ will produce, that is, $d_i:=\sum_{e \in E_i}\max\{t \mid p_{i,e}(t)=0\}.$
	Then, for every player $i$ we introduce a special resource $e_i$, and we define the set of allowable resources for this player as: $\tilde{E}_i=E_i \cup \{e_i\}$ with $e_i \neq e_j$ for $i \neq j$. The cost on these special resources $e_i$ are defined as: $c_{i,e_i}(t):=c_i(t-2d_i)$ for all $i \in N$, which are affine and increasing, and hence differentiable, convex and Lipschitz continuous with Lipschitz constant $c_i$. The cost on resources $e \in E_i$ are defined as $c_{i,e}(t):=-p_{i,e}(t)$ for all $i \in N$. Note that all cost functions are Lipschitz continuous with constant $L':=\max\{\{L\} \cup \{c_i\}_{i \in N}\}.$ In order to guarantee that all cost functions are non-negative, one can add a large positive constant $C$ to every cost function. Note that adding $C$ to every cost function does not change the equilibrium, it only adds $Cd_i$ to the total cost of each player. For each $i \in N$, we define the bijective function $\phi_i:E_i \rightarrow \tilde{E}$ as: 
	$$\phi_i(x_{i,1}, \dots, x_{i,m}) = (x_{i,1}, \dots, x_{i,m},d_i - \textstyle\sum_{e \in E_i}x_{i,e}) =: (x'_{i,1}, \dots, x'_{i,m}, x'_{i,m+1}).$$
One can check that the revenue $u_i((x_i)_{i \in N})$ for player $i$ in $\M$ equals the cost $\pi_i((\phi_i(x_i))_{i \in N}$ of $i$ in $\G$ minus a constant. Thus, game $\M$ and $\G$ are isomorphic. 
\end{proof}
As we are able to construct $\epsilon$-approximate equilibria within pseudo-polynomial time, Theorem~\ref{thrm:contcournot} implies the following theorem:

\begin{theorem}
 Given a multimarket oligopoly $\M$, where all cost functions are Lipschitz continuous with Lipschitz constant $L$, one can compute an $\epsilon$ approximate equilibrium within a running time that is pseudo-polynomial in $L$, $\max\{d_i \mid i \in N \}$ and $\tfrac{1}{\epsilon}$.
\end{theorem}

\bibliographystyle{abbrv}
\bibliography{masterbib}

\end{document}